\documentclass[a4paper,conference]{IEEEtran}
\usepackage{amsmath}%
\usepackage{mathtools} 
\usepackage{amsfonts}%
\usepackage{dsfont}%
\usepackage{amssymb}%
\usepackage{graphicx}
\usepackage{cite}
\usepackage{verbatim}%
\usepackage{enumerate}%
\usepackage[hmargin=0.65in , vmargin=1in]{geometry}
\usepackage{amsthm}
\usepackage{capt-of}
\usepackage{algorithm}
\usepackage{algorithmic}
\newtheorem{thm}{Theorem}
\newtheorem{lma}{Lemma}
\newtheorem{Def}{Definition}
\DeclareMathOperator{\E}{\mathbb{E}}

\newcommand{\lb}{\left (}
\newcommand{\rb}{\right )}

\newcommand{\script}[1]{{\mathcal {#1}}}

\newcommand{\Iavg}{I_{\rm avg}}
\newcommand{\gmax}{g_{\rm max}}

\newcommand{\fgi}{f_{g_i}}
\newcommand{\EE}[1]{\E \left[ #1 \right]}
\newcommand{\EEU}[1]{\E_{\bfU(t)} \left[ #1 \right]}

\newcommand{\bgi}{\overline{g}_i}
\newcommand{\bW}{\overline{W}}

\newcommand{\Pit}{P_i(t)}
\newcommand{\git}{g_i(t)}
\newcommand{\bfP}{{\bf P}}

\newcommand{\bgammai}{\overline{\gamma}_i}

\newcommand{\bfQ}{{\bf Q}}

\newcommand{\Xvq}{\{X(t)\}_{t=0}^\infty}
\newcommand{\Yivq}{\{Y_i(t)\}_{t=0}^\infty}
\newcommand{\Qiq}{\{Q_i(t)\}_{t=0}^\infty}

\newcommand{\bfY}{{\bf Y}}
\newcommand{\bfU}{{\bf U}}

\newcommand{\pardef}[1]{\triangleq [#1_1^{(t)},\cdots,#1_N^{(t)}]^T}
\newcommand{\parFdef}[1]{\triangleq [#1_1(t),\cdots,#1_N(t)]^T}

\newcommand{\Ri}{R_i{(t)}}

\newcommand{\gammait}{\gamma_i(t)}

\newcommand{\Rmax}{R_{\rm max}}
\newcommand{\Amax}{A_{\rm max}}
\newcommand{\gammamax}{\gamma_{\rm max}}

\newcommand{\Ts}{T_{\rm s}}

\newenvironment{prf}{\par{\noindent \it Proof:}}{\qed\par}
\ifCLASSINFOpdf
\else

\fi

\title{Dynamic Scheduling for Delay Guarantees for Heterogeneous Cognitive Radio Users}

\author{Ahmed Ewaisha and Cihan Tepedelenlio\u{g}lu\\
\small{School of Electrical, Computer, and Energy Engineering, Arizona State University, USA.}\\
\small{Email:\{ewaisha, cihan\}@asu.edu}}

\begin{document}

\maketitle

\begin{abstract}
We study an uplink multi secondary user (SU) system having statistical delay constraints, and an average interference constraint to the primary user (PU). SUs with heterogeneous interference channel statistics, to the PU, experience heterogeneous delay performances since SUs causing low interference are scheduled more frequently than those causing high interference. We propose a scheduling algorithm that can provide arbitrary average delay guarantees to SUs irrespective of their statistical channel qualities. We derive the algorithm using the Lyapunov technique and show that it yields bounded queues and satisfy the interference constraints. Using simulations, we show its superiority over the Max-Weight algorithm.
\end{abstract}

\section{Introduction}
\label{Introduction}
The problem of scarcity in the spectrum band has led to a wide interest in cognitive radio (CR) networks. CRs refer to devices that are capable of dynamically adjusting their transmission parameters according to the environment without causing harmful interference to the surrounding existing primary users (PU).

In real-time applications, such as audio and video conference calls, one of the most effective QoS metrics is the average time a packet spends in the queue before being transmitted, quantified by average queuing delay. The average queuing delay needs to be as small as possible to prevent jitter and to guarantee acceptable QoS for these applications \cite{shakkottai2002scheduling,kang2013performance}. Queuing delay has gained strong attention recently and scheduling algorithms have been proposed to guarantee small delay in wireless networks (see e.g., \cite{asadi2013survey} for a survey on scheduling algorithms in wireless systems). In the context of CR systems among the references that discuss the scheduling are \cite{Letaief_PU_Known_Location,NEP_Distributed,Ewaisha_TVT2015,Iter_Bit_Allocation_OFDM,6464638,Neely_CNC_2009}. An uplink CR system is considered in \cite{Letaief_PU_Known_Location} where the authors propose a scheduling algorithm that minimizes the interference to the PU where all users' locations including the PU's are known to the secondary base station. In \cite{NEP_Distributed} a distributed scheduling algorithm that uses an on-off rate adaptation scheme is proposed. The work in \cite{Neely_CNC_2009} proposes a scheduling algorithm to maximize the capacity region subject to a collision constraint on the PUs. The algorithms proposed in all these works aim at optimizing the throughput for the secondary users (SUs) while protecting the PUs from interference. However, providing guarantees on the queuing delay in CR systems was not the goal of these works.

The fading nature of the wireless channel requires adapting the user's rate according to the channel's fading coefficient. Many existing works on scheduling algorithms consider two-state on-off wireless channels and do not consider multiple fading levels. Among the relevant references that consider a more general fading channel model are \cite{neely2003power} and \cite{Min_Pow_4_Delay_NonCR} which did not include an average interference constraint as well as \cite{Fading_No_Scheduling,E_Hossain_CR_Delay_Analysis} where the optimization over the scheduling algorithm was not considered.

Perhaps the closest to our work are \cite{Neely_CNC_2009,li2011delay}. In \cite{Neely_CNC_2009}, the authors propose an algorithm that guarantees that the probability of collision with the PU kept below an acceptable threshold but does not give guarantees to the delay performance. The authors in \cite{li2011delay} propose a scheduling algorithm that yields an acceptable average delay performance for each user. However, in order to guarantee that the interference constraint is satisfied as well, they propose a power allocation algorithm which might not be applicable in low-cost transmitters as wireless sensor devices.

In this work we propose a scheduling algorithm that can provide delay guarantees to the SUs and protect the PUs at the same time. We show that conventional existing algorithms as the max-weight scheduling algorithm, if applied directly, can degrade the quality of service of both SUs as well as the PUs. The challenges of this problem lie in the interference constraint where it is required to protect the PU although the SUs are not capable of changing their transmission power levels. Moreover, the statistical heterogeneity of the channels might cause undesired performance for the SUs. This is because SUs located physically closer to the PUs might suffer from larger delays because closer SUs are scheduled less frequently. The SUs should be scheduled in such a way that prevents harmful interference to the PUs since they share the same spectrum. The main contribution of this paper is to propose a scheduling algorithm that satisfies both the average interference and average delay constraints.



The rest of the paper is organized as follows. The network model and the underlying assumptions are presented in Section \ref{Model}. In Section \ref{Prob_Statement} we formulate the problem mathematically. The proposed algorithm and its optimality are presented in Section \ref{Proposed_Algorithm}. Section \ref{Results} presents our simulation results. The paper is concluded in Section \ref{Conclusion}.

\section{System Model}
\label{Model}
\subsection{Channel Model}
We assume a CR system consisting of a single secondary base station (BS) serving $N$ SUs in the uplink and a single PU having access to a single frequency channel. The users are indexed according to the set $\script{N}\triangleq\{1,\cdots,N\}$. Time is divided into slots with duration $\Ts$. The PU is assumed to use the channel each time slot with probability 1. The channel between SU $i$ and the BS is referred to as direct channel $i$, while that between SU $i$ and the PU is referred to as interference channel $i$. Direct and interference channels $i$, at slot $t$, have states $\gamma_i(t)\in[0,\gammamax]$ and $g_i(t)\in[0,\gmax]$, and they follow some probability density functions $f_{\gamma_i}(\gamma)$ and $\fgi(g)$ with means $\bgammai$ and $\bgi$, respectively. Channels' states $\gammait$ and $\git$, $\forall i\in \script{N}$, are known to the BS at the beginning of time slot $t$. The channel estimation to acquire $g_i^{(t)}$ can be done by overhearing the pilots transmitted by the primary receiver, when it is acting as a transmitter, to its intended transmitter. The channel estimation phase is out of the scope of this work and the reader is referred to \cite[Section VI]{haykin2005cognitive} and \cite{Bari13ciss,Bari13asilomar,Bari14asilomar} for details on channel estimation in CRs. Direct and interference channel states are assumed to be independent and identically distributed across time slots while independent across SUs but not necessary identically distributed.

\subsection{Queuing Model}
At the beginning of each time slot $t$ packets arrive at SU $i$'s buffer with rate $\lambda_i$ packets per slot, with maximum number of arrivals $\Amax$. All packets in the system have the same length of $L$ bits. In a practical scenario, depending on the application, $L$ might be relatively small (audio packets), or relatively large (e.g. video packets). In the former, more than one packet can fit in one time slot, while in the latter a single packet might need more than one time slot for transmission \cite[Section 3.1.6.1]{semiconductor2008long}. Although in this paper we focus on the case of small $L$, our model can work for the other case as well. We assume that the buffer sizes are infinite and packets arriving to the buffer are served according to the first-come-first-serve discipline. The number of packets at SU $i$'s buffer at the beginning of slot $t$ is $Q_i(t)$ that is governed by
\begin{equation}
Q(t+1)=\lb Q(t)+\vert \script{A}_i(t)\vert-\vert \script{D}_i(t)\vert \rb^+,
\label{Queue}
\end{equation}
where $x^+\triangleq\max(x,0)$ while the set $\script{A}_i(t)$ (the set $\script{D}_i(t)$) is the set contained the indices of the packets arriving to (departing from) user $i$'s buffer at the beginning (end) of slot $t$. Define the delay of packet $j$ as $W_i^{(j)}$ which is the number of slots packet $j$ has spent in the system from the time it arrives to SU $i$'s buffer until the time it is transmitted, including the transmission time slot. $W_i^{(j)}$ has a time average $\bW_i$ that is a depends on the scheduling algorithm by which SUs are scheduled. $\bW_i$ is given by
\begin{equation}
\bW_i \triangleq \limsup_{T \rightarrow \infty}\frac{\EE{\sum_{t=0}^T{\sum_{j\in\script{A}_i(t)}W^{(j)}_i(t)}}}{\EE{\sum_{t=0}^T{\vert \script{A}_i(t)\vert}}}.
\label{Delay}
\end{equation}

\subsection{Transmission Process}
At the beginning of time slot $t$, the BS chooses a user, say user $i$, according to some scheduling algorithm. Define the vector $\bfP(t)\pardef{P}$ where $P_i(t)=1$ if SU $i$ is allocated the channel at time slot $t$ and $0$ otherwise. If $P_i(t)=1$, SU $i$ adapts its transmission rate according to the channel's gain and begins transmission its packets. Thus the number of packets transmitted is
\begin{equation}
\Ri=\log \lb 1+P_i(t)\gamma_i{(t)} \rb \hspace{0.1in} \rm{packets},
\label{Tx_Rate}
\end{equation}
with a maximum rate $\Rmax\triangleq\log \lb 1+\gammamax\rb$. At the end of this time slot, the BS receives these packets error-free and then slot $t+1$ begins.

\section{Problem Statement}
\label{Prob_Statement}
Each SU $i$ has an average delay constraint $\bW_i \leq d_i$ that needs to be satisfied, where $d_i$ is the maximum average delay that SU $i$ can tolerate. Moreover, the PU can tolerate a maximum interference of $\Iavg$ aggregated over the interference received from all SUs. Define $\script{T}$ as the set containing the indices of the time slots where there is at least one SU having at least one packet in its queue, or $\script{T}\triangleq \{t:t\geq 1, \sum_{i=1}^N{Q_i(t)}>0\}$. The main objective is to find a scheduling algorithm that guarantees the SUs' delay constraints as well as the PU's interference constraint. That is, find the value of $\bfP(t)$ at each slot $t$ that solves
\begin{equation}
\begin{array}{ll}
\underset{\{\bfP(t)\}}{\rm{minimize}}& 0
\label{Problem}\\
\rm{subject \; to} & I\triangleq \lim_{T\rightarrow \infty} \sum_{i=1}^N{\frac{1}{T}\sum_{t=1}^T{P_i(t) \git}} \leq \Iavg\\
& \bW_i \leq d_i , \hspace{0.1in} \forall i\in\script{N}\\
&\EE{\vert\script{A}_i(t)\vert}\leq\EE{\vert\script{D}_i(t)\vert}, \hspace{0.1in} \forall i\in\script{N}\\
& \sum_{i=1}^N{ P_i(t)} \leq 1 \hspace{0.25in}, \hspace{0.1in} \forall t\in \script{T}
\end{array}
\end{equation}
where $I$ is the time-averaged interference affecting the PU due to the SUs' transmissions.

We notice that the constraints of problem \eqref{Problem} are expressed in terms of asymptotic time averages and cannot be solved by conventional optimization techniques. The next section proposes a low complexity update algorithm and proves its convergence to a feasible point solving \eqref{Problem}.

\section{Solution Approach}
\label{Proposed_Algorithm}
We propose an online algorithm that schedules the users at slot $t$ based on the history up to slot $t$. We show that this algorithm has an optimal performance.

\subsection{Satisfying Delay Constraints}
In order to satisfy the delay constraints in problem \eqref{Problem}, we set up a ``virtual queue'' associated with each delay constraint $\bW_i\leq d_i$. The virtual queue for SU $i$ at slot $t$ is given by
\begin{equation}
Y_i(t+1)\triangleq\lb Y_i(t)+\sum_{j\in \script{D}_i(t)}{\lb W_i^{(j)}-d_i\rb} \rb^+
\label{Delay_Q}
\end{equation}
where we initialize $Y_i(0)\triangleq 0$, $\forall i$. We define $\bfY(t) \parFdef{Y}$. Equation \eqref{Delay_Q} is calculated at the end of slot $t-1$ and represents the amount of delay exceeding the delay bound $d_i$ for SU $i$ up to the beginning of slot $t$. We first give the following definition, then state a lemma that gives a sufficient condition on $Y_i(t)$ for the delay of SU $i$ to satisfy $\bW_i \leq d_i$.
\begin{Def}
\label{Mean_Rate_Def}
A random sequence $\{Y_i(t)\}_{t=0}^\infty$ is mean rate stable if and only if the equality $\lim_{T\rightarrow\infty}\EE{Y_i(T)}/T=0$ holds.
\end{Def}

\begin{lma}
\label{Mean_Rate_Lemma}
If the arrival rate vector ${\bf \lambda}\triangleq[\lambda_1,\cdots\lambda_N]^T$ can be supported, i.e. the number of arrivals equals the number of departures over a large period of time, and if $\{Y_i(t)\}_{t=0}^\infty$ is mean rate stable, then the time-average delay of SU $i$ satisfies $\bW_i \leq d_i$.
\end{lma}
\begin{proof}
Removing the $(\cdot)^+$ sign from equation \eqref{Delay_Q} yields
\begin{equation}
Y_i(t+1) \geq Y_i(t)+\sum_{j\in \script{D}_i(t)}{\lb W_i^{(j)}-d_i\rb}.
\label{Inequality1}
\end{equation}
Summing inequality \eqref{Inequality1} over $t=0,\cdots T-1$ and noting that $Y_i(0)=0$ gives
\begin{equation}
Y_i(T)\geq \sum_{t=0}^{T-1} \lb\sum_{j\in \script{D}_i(t)}W_i^{(j)}\rb-d_i\sum_{t=0}^{T-1}\lb \vert D_i(t)\vert\rb.
\label{Inequality2}
\end{equation}
Taking the $\EE{\cdot}$ then dividing by $\EE{\sum_{t=0}^{T-1}{\vert\script{D}_i(t)\vert}}$ gives
\begin{equation}
\frac{\EE{\sum_{t=0}^{T-1} \lb\sum_{j\in \script{D}_i(t)}W_i^{(j)}\rb}}{\EE{\sum_{t=0}^{T-1}{\vert\script{D}_i(t)\vert}}} \leq \frac{\EE{Y_i(T)}}{\EE{\sum_{t=0}^{T-1}{\vert\script{D}_i(t)\vert}}} + d_i.
\label{Wait_r_i}
\end{equation}
Using the fact that there exists some $T\geq1$ where $\sum_{t=0}^{T-1}{\vert\script{D}_i(t)\vert}=\sum_{t=0}^{T-1}{\vert\script{A}_i(t)\vert}$, a fact based on the lemma's queue-stability assumption, taking the limit as $T\rightarrow\infty$, using the identity $\EE{\vert\script{A}_i(t)\vert}=\lambda_i$, the mean rate stability definition as well as equation \eqref{Delay} completes the proof.
\end{proof}
Lemma \ref{Mean_Rate_Lemma} provides a condition on the virtual queue $\Yivq$ so that SU $i$'s average delay constraint $\bW_i\leq d_i$ in \eqref{Problem} is satisfied. That is, if the proposed scheduling algorithm results in a mean rate stable $\Yivq$, then $\bW_i\leq d_i$.

\subsection{Satisfying Interference Constraints}

To track the average interference at the PU up to the end of slot $t$ we set up the following virtual queue that is associated with the average interference constraint in problem \eqref{Problem} and is calculated at the BS at the end of slot $t$.
\begin{equation}
X(t+1)\triangleq\lb X(t)+\sum_{i=1}^N{ P_i{(t)}g_i{(t)}}-\Iavg \rb^+
\label{Interf_VQ}
\end{equation}
where the term $\sum_{i=1}^N{ P_i{(t)}g_i{(t)}}$ represents the aggregate amount of interference energy received by the PU due to the transmission of a SU during slot $t$. Hence, this virtual queue is a measure of how much the SUs have exceeded the interference constraint above the level $\Iavg$ that the PU can tolerate. Lemma \ref{Mean_Rate_Lemma_Avg_Interf} provides a sufficient condition for the interference constraint of problem \eqref{Problem} to be satisfied.

\begin{lma}
\label{Mean_Rate_Lemma_Avg_Interf}
If $\{X(t)\}_{t=0}^\infty$ is mean rate stable, then the time-average interference received by the PU satisfies $I \leq \Iavg$.
\end{lma}
\begin{proof}
The proof is similar to that of Lemma \ref{Mean_Rate_Lemma} and is omitted for brevity.
\end{proof}
Lemma \ref{Mean_Rate_Lemma_Avg_Interf} says that if the power allocation and scheduling algorithm results in mean rate stable $\Xvq$, then the interference constraint of problem \eqref{Problem} is satisfied.

\subsection{Proposed Algorithm}

The proposed algorithm, illustrated in Algorithm \ref{Alg_1}, is executed at the beginning of each time slot to find the scheduled user. The idea is to choose the user that minimizes\
\begin{equation}
\Psi_i(t)\triangleq \sum_{i=1}^N \left[XP_ig_i+Y_i\lb\sum_{j\in\script{D}_i}\lb W_i^{(j)}- d_i\rb\rb-Q_i\vert\script{D}_i\vert\right],
\label{psi}
\end{equation}
under the constraint that $\sum_{i=1}^N P_i(t)\leq 1$, where we drop the index $t$ from all variables in \eqref{psi} for brevity. This is equivalent to scheduling the user with the smallest $\phi_i(t)$ where, after dropping the index $t$, $\phi$ is defined as
\begin{equation}
\phi_i\triangleq Xg_i+Y_i\sum_{j\in\script{D}_i}W_i^{(j)}- \lb Y_id_i+Q_i\rb R_i,
\end{equation}
or otherwise set $P_i(t)=0$ $\forall i\in\script{N}$ if $\phi_i(t)>0$ $\forall i\in \script{N}$. We now state our proposed algorithm then discuss its optimality.

\begin{algorithm}
\caption{Finding the optimum scheduling rule $\bfP(t)$ at slot $t$}
\begin{algorithmic}[1]
\label{Alg_1}
\STATE Find the set of backlogged users $\script{B}(t)\triangleq\{i:Q_i(t)>0\}$.
\STATE Set $P_{i^*}(t)=1$ where $i^* \in \arg\min_{i\in\script{B}(t)} \phi_i(t)$. Ties broken arbitrarily. Set $P_{i}(t)=0$, $\forall i\neq i^*$.
\STATE Update the variables $X(t)$ and $Y_i(t)$, $\forall i \in\script{N}$ using \eqref{Interf_VQ} and \eqref{Delay_Q}, respectively.
\end{algorithmic}
\end{algorithm}
The optimality of Algorithm \ref{Alg_1} is discussed in Theorem \ref{Optimality}.
\begin{thm}
\label{Optimality}
Under Algorithm \ref{Alg_1}, the inequality $\EE{\vert\script{A}_i(t)\vert}\leq\EE{\vert\script{D}_i(t)\vert}$ holds $\forall i\in\script{N}$ and the virtual queues $\Xvq$ and $\Yivq$ are mean rate stable.
\end{thm}
\begin{prf}
See Appendix \ref{Optimality_Proof_Inst}.
\end{prf}
Theorem \ref{Optimality} states that if the problem is feasible, then Algorithm \ref{Alg_1} results in stable virtual queues. This is achieved by balancing between scheduling a user with a large average delay up to slot $t$ but interferes more with the PU at slot $t$, and another user with relatively smaller average delay in the past but has a small gain to the PU at slot $t$. In the proof we show that this algorithm minimizes the drift of the Lyapunov function and thus guarantee that the virtual queues do not build up, indicating that the constraints are satisfied.

In problem \eqref{Problem}, the constraint $\sum_{i=1}^N P_i(t)\leq 1$ is needed to insure that no more than one user is scheduled at each time slot. We note that this constraint means that Algorithm \ref{Alg_1} might set $P_i(t)=0$, $\forall i\in\script{N}$, when $\phi_i(t)>0$, $\forall i\in\script{N}$. Hence, even if there is a packet in the system waiting for transmission, the channel might not be assigned to any user, but will remain idle. While this constraint might not yield a non-idling scheduling algorithm, it guarantees that the interference constraint is satisfied. That is, if this constraint is replaced with the following constraint: $\sum_{i=1}^N P_i(t)= 1$, the resulting algorithm is a non idling one that might not satisfy the PU's interference constraint. We elaborate more on this.

In queuing theory, a non-idling scheduling algorithm always schedules a user whenever there is a packet in the system to be transmitted. In other words, the server (wireless channel) is never left idle (unassigned to any users) unless all users have empty backlogs. Applying any non-idling scheduling algorithm to our problem, although might have better delay performance, results in the PU receiving interference whenever the users are backlogged. This interference averaged over a large period of time might exceed $\Iavg$. However, Algorithm \ref{Alg_1} assigns the channel to a user when its interference gain $\git$ is relatively low, and idles the channel when all gains are relatively high. Hence, our algorithm makes use of the interference channels' random nature and assigns the channel opportunistically to users.

\section{Simulation Results}
\label{Results}
We simulated the system for $N=2$ SUs (refer to Table \ref{Parameters} for complete list of parameter values). The system was simulated for a deterministic direct channel and a Rayleigh fading interference channel. We simulated the system until the average of the virtual queues normalized by the time are negligible, that is $(\EE{X(T)}+\sum_{i=1}^N \EE{Y_i(T)})/\lb(N+1)T\rb<\epsilon$. 

\begin{table}
\centering
		\caption{Simulation Parameter Values}
		\begin{tabular}{|c|c||c|c|}
			\cline{1-4}
			Parameter & Value & Parameter & Value\\ \cline{1-4}
			$\lambda_1=\lambda_2$ & $\lambda \in [0.02,0.4]$ & $\Iavg$ & 2 \\
			$(\bar{\gamma}_1,\bar{\gamma}_2)$ & $(1,1)$&$\epsilon$ & $0.01$ \\
			$(\bar{g}_1,\bar{g}_2)$ & $(0.4,0.2)$ & $(d_1,d_2)$ & $(1.5,5)$ \\
			\cline{1-4}
			\end{tabular}
		\label{Parameters}
		\end{table}

Fig. \ref{NonIdling_PerSU_Delay} plots the delay of each SU versus $\lambda$, where $\lambda \triangleq \lambda_1=\lambda_2$, for two different scenarios; the first being the non idling version of the proposed algorithm, that is we minimize $\Psi(t)$ subject to $\sum_iP_i(t)=1$, while the second is the Max-Weight (MW) algorithm that schedules the user with the highest $Q_i(t)/g_i(t)$. The essence of the MW algorithm lies in assigning the channel to the user who has more packets in the queue and expected to interfere less with the PU. Clearly, the MW will schedule user $2$ more frequently than user $1$ since $\bar{g}_2<\bar{g}_1$, hence the delay of SU $2$ will be less than that of user $1$. This means that the heterogeneity of the interference channels has resulted in differentiation in the service provided to the SUs to protect the PU. On the other hand, our proposed algorithm can bound SU $1$'s average delay to guarantee a fixed QoS even if its channel is worse than SU $2$. The draw back of the non idling algorithm is that the interference constraint is not guaranteed to be satisfied. This is demonstrated in Fig. \ref{NonIdling_Interf} where the average interference of both algorithms coincide.

		\begin{figure}%
			\centering
			\includegraphics[width=1\columnwidth]{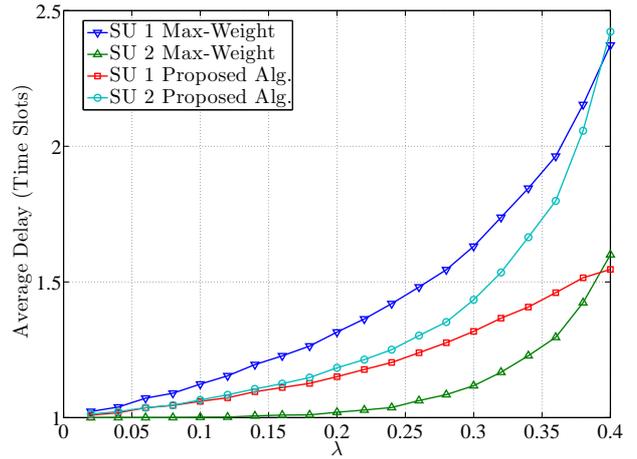}%
			\caption{Average Delay of each SU versus $\lambda_1=\lambda_2=\lambda$ for the non-idling version of Algorithm \ref{Alg_1}.}%
			\label{NonIdling_PerSU_Delay}%
		\end{figure}
		\begin{figure}%
			\centering
			\includegraphics[width=1\columnwidth]{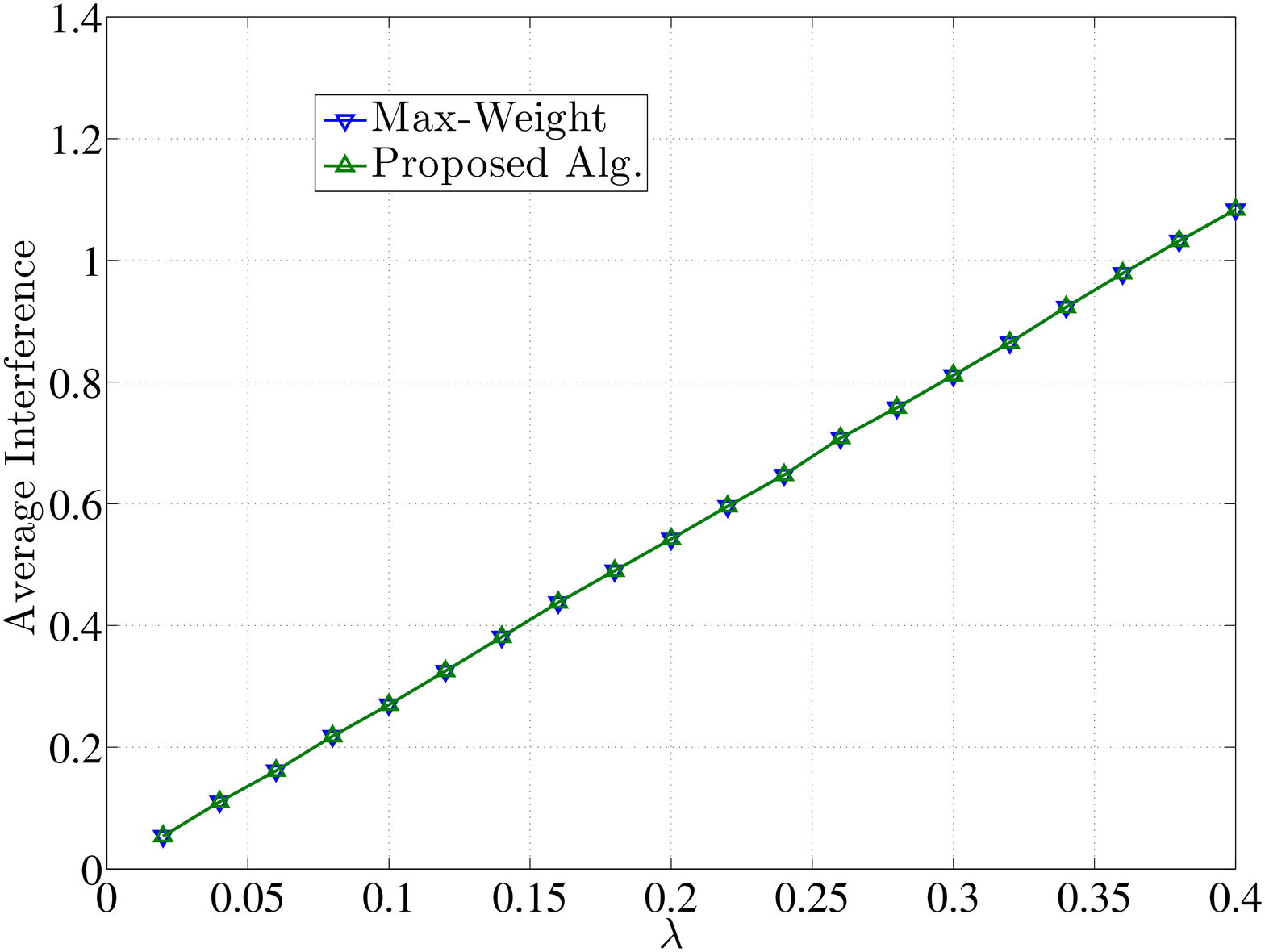}%
			\caption{Average interference to the PU versus $\lambda_1=\lambda_2=\lambda$ for the non-idling version of Algorithm \ref{Alg_1}.}%
			\label{NonIdling_Interf}%
		\end{figure}

Fig. \ref{Idling_PerSU_Delay} compares the per-SU delay performance of the Max-Weight algorithm to that of Algorithm \ref{Alg_1}. Although Algorithm \ref{Alg_1} suffers a higher sum delay since it is not a non-idling algorithm, it can bound SU $1$'s average delay to the required delay value. At the same time, the PU is protected under the proposed algorithm. This is demonstrated in Fig. \ref{Idling_Interf} where the interference suffered by the PU is less than $\Iavg$ while the Max-weight fails to protect the PU.
		\begin{figure}%
			\centering
			\includegraphics[width=1\columnwidth]{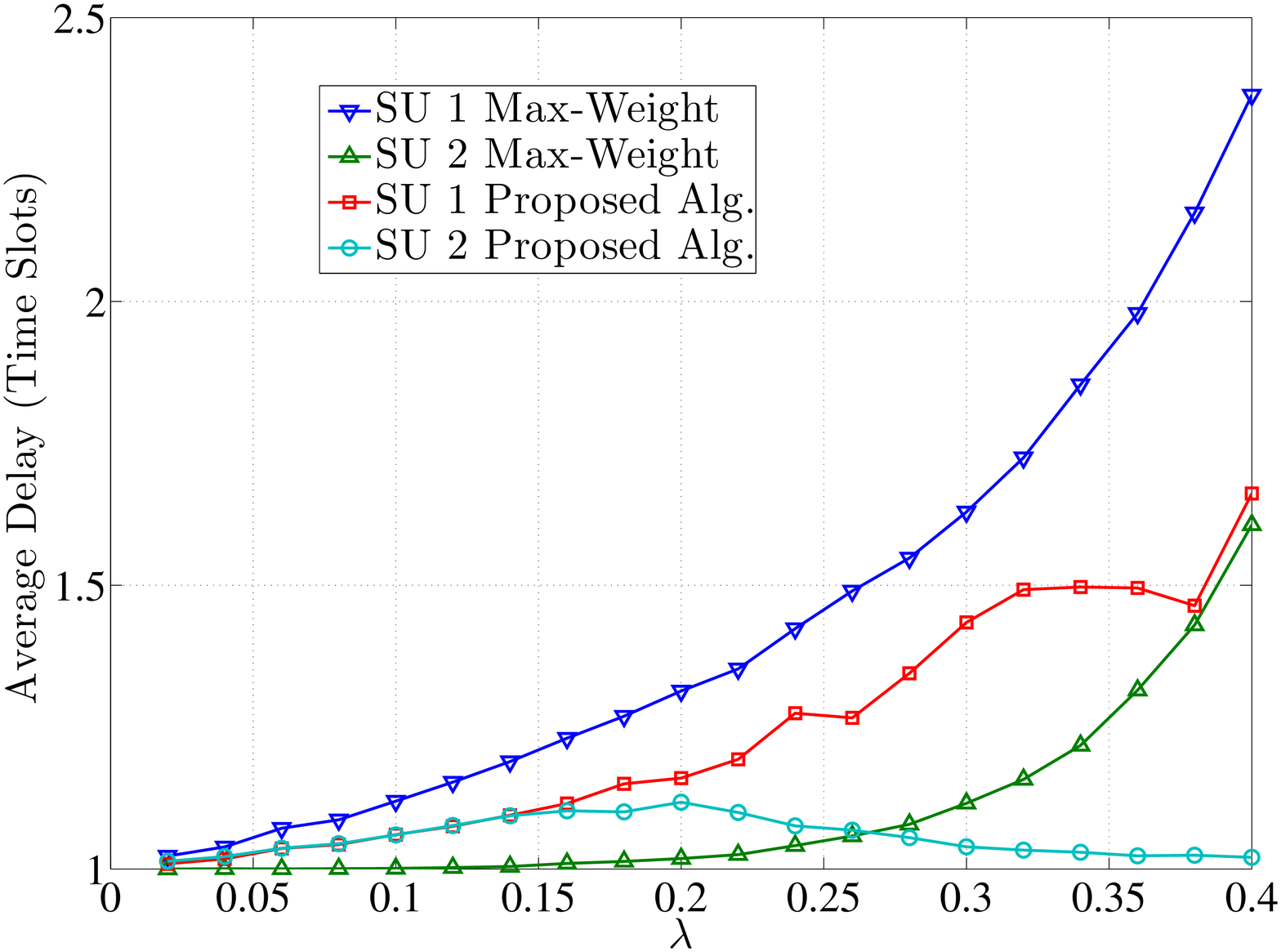}%
			\caption{Average Delay of each SU versus $\lambda_1=\lambda_2=\lambda$ for Algorithm \ref{Alg_1}. SU $1$'s average delay can be controlled using the proposed algorithm.}%
			\label{Idling_PerSU_Delay}%
		\end{figure}
		\begin{figure}%
			\centering
			\includegraphics[width=1\columnwidth]{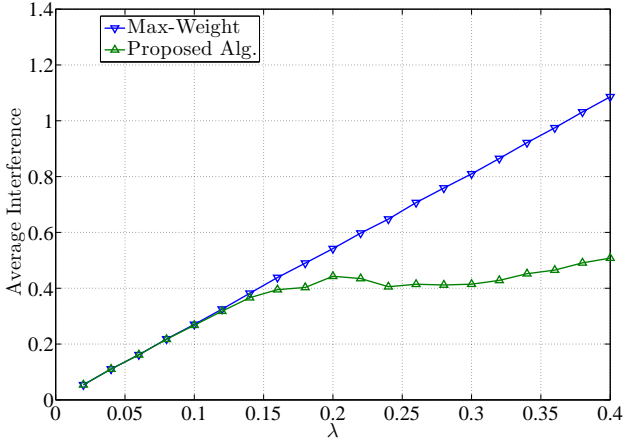}%
			\caption{Average interference to the PU versus $\lambda$ for Algorithm \ref{Alg_1} is less than that for the Max-Weight algorithm.}%
			\label{Idling_Interf}%
		\end{figure}


\section{Conclusion}
\label{Conclusion}
We have studied the scheduling problem in a multi-SU uplink system. The motivation behind this problem is that each SU has an average delay constraint that conventional algorithms cannot satisfy since they neglect the heterogeneity of the interference channels causing intolerable interference to the PU. An optimal scheduling algorithm was proposed. The provided simulation results show that the heterogeneity in the interference channels can lead to suffering of one of the SUs from high delay. The proposed algorithm dynamically reallocates the channel to the suffering SUs to decrease their delays without violating the interference constraint.

\bibliographystyle{IEEEtran}
\bibliography{MyLib}

\appendices
\section{Proof of Theorem \ref{Optimality}}
\label{Optimality_Proof_Inst}
\begin{proof}
In this proof, we show that the drift under this algorithm is upper bounded by some constant, which indicates that the virtual queues are mean rate stable \cite{georgiadis2006resource,urgaonkar2011optimal}.

We define $\bfU(t)\triangleq [X(t) , \bfY(t) , \bfQ(t)]^T$, the Lyapunov function as $L(t) \triangleq \frac{1}{2}X^2(t)+\frac{1}{2}\sum_{i=1}^N \lb Y_i^2(t) + Q_i^2(t)\rb$ and Lyapunov drift to be $\Delta (t) \triangleq \EEU{L(t+1) - L(t)}$ where $\EEU{x}\triangleq\EE{x\vert \bfU(t)}$. Squaring \eqref{Queue}, \eqref{Delay_Q} and \eqref{Interf_VQ} then taking the conditional expectation we can get the bounds
\begin{align}
\nonumber \frac{1}{2}&\E_{\bfU(t)}\left[Q_i^2(t+1)-Q_i^2(t)\right] \leq\\
&Q_i(t)\EEU{\vert \script{A}_i(t)\vert -\vert \script{D}_i(t)\vert} + C_{Q_i},
\label{Queue_Q_Sq2}\\
\nonumber \frac{1}{2}&\E_{\bfU(t)} \left[Y_i^2(t+1)-Y_i^2(t)\right] \leq\\
&Y_i(t)\EEU{\sum_{j\in \script{D}_i(t)}\lb W_i^{(j)}-d_i\rb} + C_{Y_i}, \hspace{0.1in} {\rm and}
\label{Delay_Q_Sq2}\\
\nonumber\frac{1}{2}&\E_{\bfU(t)} \left[X^2(t+1)-X^2(t)\right] \leq \\
&C_X+X(t)\lb\EEU{\sum_{i=1}^N\Pit \git}-\Iavg\rb,
\label{Interf_Q_Sq1}
\end{align}
respectively, where we use the bounds $\lb\sum_{i=1}^N\Pit \git\rb^2+\Iavg^2<C_X$, $\EEU{\vert \script{A}_i(t)\vert^2} +\EEU{\vert \script{D}_i(t)\vert^2}\leq C_{Q_i}$ and
\begin{align*}
& d_i^2\EEU{\vert \script{D}_i(t)\vert^2}+\EEU{\lb\sum_{j\in \script{D}_i(t)} W_i^{(j)}\rb^2}<C_{Y_i}
\end{align*}
in \eqref{Interf_Q_Sq1} where $C_{Q_i}\triangleq \Amax^2+\Rmax^2$ and $C_X\triangleq \gmax^2+\Iavg^2$. We omit the derivation of these bounds for brevity. The derivation is similar to that in \cite[Lemma7]{li2011delay}. Using the bounds in \eqref{Queue_Q_Sq2}, \eqref{Delay_Q_Sq2} and \eqref{Interf_Q_Sq1}, the drift becomes bounded by $\Delta\lb \bfU(t)\rb\leq C+\EEU{\Psi(t)}$, where $C\triangleq C_X+\sum_{i=1}^N\lb C_{Y_i}+C_{Q_i}\rb$. Now, since $\Psi(t)<0$ under Algorithm \ref{Alg_1}, then $\Delta(t)\leq C$. Taking $\EE{\cdot}$, summing over $t=0,\cdots,T-1$, denoting $X(0)\triangleq\bfY_i(0)\triangleq 0$ for all $i\in\script{N}$, and dividing by $T$ we get $\frac{\EE{X^2(T)}}{T}+\sum_{i=1}^N \frac{\EE{Q^2_i(T)+Y_i^2(T)}}{T}\leq C$. Removing all the terms on the left-hand-side of the last inequality except the term $Q^2_i(T)/T$ we obtain $\EE{Q_i^2(T)}/T \leq C$. Using Jensen's inequality we note that
\begin{equation}
\frac{\EE{Q_i(T)}}{T} \leq \sqrt{\frac{\EE{Q_i^2(T)}}{T^2}} \leq \sqrt{\frac{C}{T}}.
\label{Jensens}
\end{equation}
Finally, taking the limit when $T\rightarrow \infty$ completes the mean rate stability proof of $\Qiq$, which means that $\EE{\vert\script{A}_i(t)\vert}\leq\EE{\vert\script{D}_i(t)\vert}$. The proofs of the mean rate stability of $\Xvq$ and $\Yivq$ follow similarly.
\end{proof}

\end{document}